\documentclass[jmp,amsmath,amssymb,floatfix]{revtex4-1}

\usepackage{mathtools}
\usepackage{mathrsfs}
\usepackage{mathtools} 
\mathtoolsset{showonlyrefs,showmanualtags}

\usepackage{graphicx}
\usepackage{subcaption}
\usepackage{mwe}

\usepackage{xcolor}
\usepackage{amsthm}
\numberwithin{equation}{section}
\newtheorem{thm}{Theorem}[section]
\newtheorem{lemma}[thm]{Lemma}
\newtheorem{cor}[thm]{Corollary}
\newtheorem{rem}[thm]{Remark}

\graphicspath{{figs-pdf/}}

\newcommand{\C}{\mathbb{C}}

\newcommand{\R}{\mathbb{R}}

\newcommand{\be}{\begin{equation}}
\newcommand{\ee}{\end{equation}}
\newcommand{\ba}{\begin{align}}
\newcommand{\ea}{\end{align}}
\newcommand{\ben}{\begin{equation*}}
\newcommand{\een}{\end{equation*}} 
\newcommand{\bt}{\begin{thm}}
\newcommand{\et}{\end{thm}}
\newcommand{\bl}{\begin{lem}}
\newcommand{\el}{\end{lem}}
\newcommand{\bc}{\begin{cor}}
\newcommand{\ec}{\end{cor}}
\newcommand{\br}{\begin{rem}}
\newcommand{\er}{\end{rem}}
\newcommand{\e}{\varepsilon}
\newcommand{\thh}{^\text{th}}

\begin{document}

\title{Multiplane gravitational lenses with an abundance of images}

\author{Charles R. Keeton}
\affiliation{
Department of Physics and Astronomy, Rutgers University, 136 Frelinghuysen Road, Piscataway, NJ 08854
}

\author{Erik Lundberg}
\affiliation{
Department of Mathematical Sciences, Florida Atlantic University, \\
777 Glades Rd., Boca Raton, FL 33431\\
elundber@fau.edu (email address of corresponding author)
}

\author{Sean Perry}
\affiliation{Department of Mathematics \& Statistics, University of South Florida, \\
4202 E Fowler Ave, Tampa, FL 33620
}

\date{\today}

\begin{abstract}
We consider gravitational lensing of a background source by a finite system of point-masses.
The problem of determining the maximum possible number of lensed images has been completely resolved in the single-plane setting (where the point masses all reside in a single lens plane), but this problem remains open in the multiplane setting.  
We construct examples of $K$-plane point-mass gravitational lens ensembles that produce 
$\prod_{i=1}^K (5g_i-5)$
images of a single background source, 
where $g_i$ is the number of point masses in the $i^\text{th}$ plane.  This gives asymptotically (for large $g_i$ with $K$ fixed) $5^K$ times the minimal number of lensed images. Our construction uses Rhie's single-plane examples and a structured parameter-rescaling algorithm to produce preliminary systems of equations with the desired number of solutions. Utilizing the stability principle from differential topology, we then show that the preliminary (nonphysical) examples can be perturbed to produce physically meaningful examples while preserving the number of solutions.
We provide numerical simulations illustrating the result of our construction, including the positions of lensed images as well as the structure of the critical curves and caustics.
We observe an interesting ``caustic of multiplicity'' phenomenon that occurs in the nonphysical case and has a noticeable effect on the caustic structure in the physically meaningful perturbative case.
\end{abstract}

\maketitle

\section{Introduction}

Gravitational lensing occurs when the gravity due to massive objects acts as a lens, bending light from a background source.  Besides magnifying or distorting light from the source, a gravitational lens can produce multiple images of a single source.  The possibility of multiple images naturally leads to the problem of determining (within a given class of mathematical models) possible numbers of images that may be lensed.  While this problem is simple to state, it has proven to be quite challenging.  In this paper we focus on the particular setting of lensing by point masses.

Einstein noticed \cite{Einstein} that a single point-mass will generically produce two images of a background source, and Schneider and Weiss \cite{SW1986} showed that a pair of point masses residing in a common lens plane (orthogonal to the observer's line of sight) can produce three or five images.

For single-plane lensing by $g$ point masses, Petters established \cite{Pet97} an upper bound that increases quadratically in $g$.  Mao, Petters, and Witt \cite{MPW} conjectured that the maximum actually increases linearly in $g$, and they produced configurations that produce $3g+1$ images. S.\ H.\ Rhie constructed configurations that yield $5g-5$ images \cite{Rhie}, and she conjectured that those examples are extremal.  Khavinson and Neumann \cite{KN} confirmed her conjecture, using an indirect method based on holomorphic dynamics to show that there can be at most $5g-5$ images of a single background source.  In light of Rhie's examples, the Khavinson-Neumann bound is sharp, i.e., $5g-5$ is the \emph{maximum} number of lensed images. We state this result as a theorem.

\begin{thm}[Khavinson, Neumann, Rhie]
\label{thm:single}
For single-plane lensing by $g$ point masses, the maximum number of lensed images is $5g-5$.
\end{thm}

Concerning the minimum number of images for single-plane lensing, Petters \cite{Pet92} used Morse theory to show that there are always at least $g+1$ images, and it is easy to construct examples that attain this lower bound (simply moving the source sufficiently far from the origin in the source plane).

More generally, Petters determined the minimum number of images in the multiplane setting. Let $N$ denote the number of images lensed by a $K$-plane gravitational lens with $g_i$ masses in the $i$th plane.  Petters proved the lower bound $N \geq \prod_{i=1}^K (g_i+1)$, and again this is indeed the minimum as there are examples that have exactly this number of lensed images (again by moving the source far from the origin in the source plane).

On the other hand, as Petters pointed out in his survey paper \cite{Pet10}, there is yet no multiplane analog of Theorem \ref{thm:single}.  This brings us to the following problem that remains open.

\vspace{\baselineskip}
\noindent {\bf Problem.}  Determine the maximum number of images for multiplane lensing, i.e., given positive integers $g_1,g_2,...,g_K$, determine the maximum number of images that can be lensed by a $K$-plane point-mass system with $g_i$ masses in the $i\thh$ plane.
\vspace{\baselineskip}

The third named author of the current paper recently established the upper bound \cite{Perry}
\be\label{eq:Sean}
N \leq E_K^2 + O_K^2
\ee
where $E_K$ and $O_K$ denote the sums of the coefficients of the even and odd degree terms respectively in the formal polynomial  $\prod_{i=1}^K (1+g_i Z)$.
Petters \cite{Pet95} had previously proved the upper bound $N \leq 2(2^{2(K-1)}-1)$ in the special case when there is a single point mass in each of the $K$ lens planes.

The estimate \eqref{eq:Sean} increases quadratically in each $g_i$.   Motivated by the outcome for the single-plane case stated in Theorem \ref{thm:single}, it was asked in \cite{Perry} whether the estimate \eqref{eq:Sean} can be improved to a bound that is linear in each $g_i$.
More specifically, it was asked \cite[Concluding Remarks]{Perry} whether the estimate $N \leq \prod_{i=1}^K (5g_i-5)$ holds when $g_i \geq 2$.


The main goal of the current paper is to construct examples that produce $\prod_{i=1}^K (5g_i-5)$ lensed images.  Hence, if the above bound does hold then it is best possible.

\begin{rem}
Note that the term ``image'' in the context of ``lensed image'' carries the physical meaning referring to what the observer would see (with the aid of a telescope).  This is in direct opposition to mathematical meaning, in fact, the positions of the (physical) images occur at (mathematical) pre-images of the background source under the lensing map (the lensing map is discussed in Section \ref{sec:numerics}).
\end{rem}


Let us now formulate the multiplane lensing model in the form of a system of lensing equations.
We recall from  \cite[pg 199]{PetBook} that the lensed images correspond to solutions $(x_1, ... , x_K) \in \R^{2K}$ of the system of equations
\be\label{eq:systemgen}
\begin{cases}
x_2   &= x_1 -\beta_1 \alpha_1 (x_1)  \\
x_{i+1}  &= x_i + \e_i (x_i - x_{i-1}) - \beta_i \alpha_i(x_i), \quad i=2,3,...,K
\end{cases}
\ee
where $\beta_{i}>0$ and $\e_i>0$ are scaling constants derived from the distances between planes (see the Appendix, Section \ref{sec:param}), $x_{K+1} = y \in \R^2$ is the (fixed) location of the background source in the source plane (orthogonal to the observer's line of sight), and $\alpha_j$ is the bending angle vector of the $j^\text{th}$ plane which can be expressed in terms of position and mass parameters as
\be\label{eq:alpha}
\alpha_j(x_j)=\sum_{\ell=1}^{g_j} b^2_{j,\ell} \frac{x_j-\xi_{j,\ell}}{|x_j-\xi_{j,\ell}|^2},
\ee
where $b_{j,\ell}$ is the Einstein radius of the $\ell\thh$ point mass positioned at $\xi_{j,\ell}$ in the $j\thh$ plane (the square $b_{j,\ell}^2$ of the Einstein radius corresponds to mass).
A solution to system \eqref{eq:systemgen} may be viewed as a list $(x_1,x_2,...,x_K)$ of locations $x_j$ where a light ray, traced backward from observer to source, impacts the $j\thh$ plane. Note that the choice of some $x_1$ determines the value of each $x_j$ where $1\leq j \leq K$. Should any such $x_j = \xi_{j, \ell}$ for any $\ell$, the point $x_1$ is referred to as an \textit{obstruction point}. These correspond to light rays traced back from the observer which impact a lensing mass at some point. They also correspond to those values of $x_i$ which would determine a $x_j$ that caused $\alpha_j(x_j)$ to be undefined (see Figure \ref{fig:Multiplane}). With this framework, we can now state precisely our main result.

\begin{figure}
    \includegraphics[width=1\textwidth]{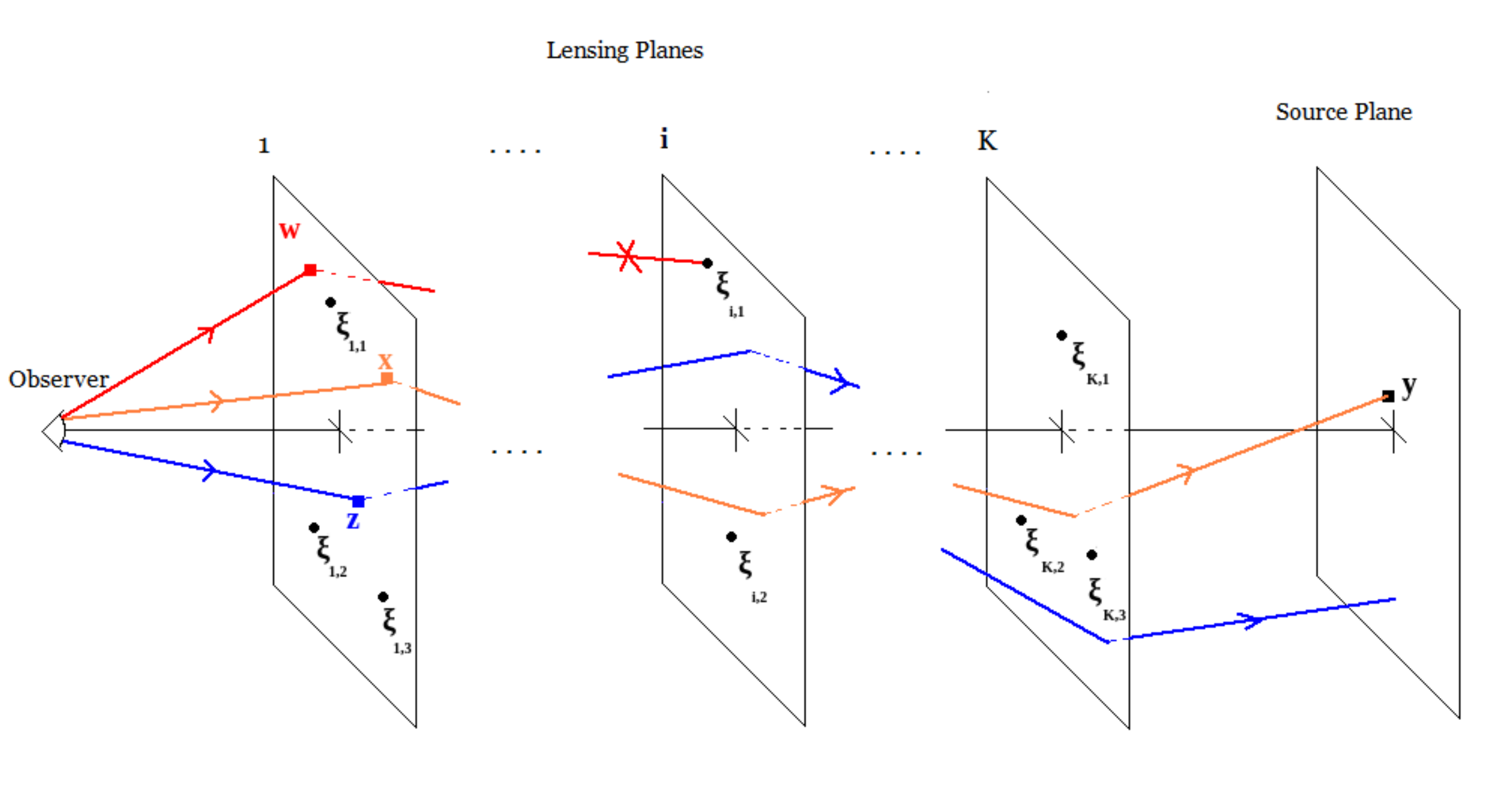}
\caption{Multiplane lensing by point masses.
Here, $x_1 = x$ is a solution to the lensing system, $x_1 = z$ is not, and $x_1 = w$ indicates an obstruction point. The points $\xi_{i,j}$ indicate the $j\thh$ mass on the $i\thh$ plane.}\label{fig:Multiplane}
\end{figure}

\begin{thm}\label{thm:examples}
For each list $g_1,g_2,...,g_K$ of integers $g_i \geq 2$, there exists a choice of parameters for which the system of equations \eqref{eq:systemgen} is nondegenerate and has $\prod_{i=1}^K (5g_i-5)$ solutions.
In other words, the corresponding
$K$-plane gravitational lens, with $g_i$ point masses in the $i^\text{th}$ plane, produces
\be\label{eq:abundance}
\prod_{i=1}^K (5g_i-5)
\ee
nondegenerate lensed images of a single background source.
\end{thm}

Comparing this with the result of Petters on the minimal number of images, we notice that the quantity \eqref{eq:abundance} has an additional factor of $5^K$.  We deem this an ``abundance'' of images, however, we should reiterate that the problem of determining the maximum number of images for multiplane lensing remains open.

\subsection*{Related work}

In addition to the above-mentioned studies on the image counting problem for point-mass lenses, 
let us briefly mention some results in relevant directions (we also point the reader to surveys of results up to 2010 that can be found in \cite{Pet10}, \cite{PettersWerner}).
Following the breakthrough of Khavinson and Neumann that was mentioned above, single-plane lensing by point masses has been investigated further in
\cite{Bleher}, \cite{SeteCMFT}, \cite{SeteGrav}, \cite{SetePert}, \cite{Zur2018a}, \cite{Zur2018b},
\cite{Kuznia}.

The number of lensed images is related to the study of caustic structure, a topic that has been investigated in some of the previously mentioned papers as well as in \cite{PettersWittCaustics}, \cite{AazamiPettersRabin}, \cite{BozzaCusps}, and  \cite{PettersCusps}.  We note that in \cite{PettersCusps}, an upper bound on the number of cusps on the lensing caustic is established for single-plane lensing by $g$ point masses.  The upper bound grows quadratically in $g$, and it is asked in \cite{Pet10} whether this can be improved to a linear bound.  So far, there has apparently been no progress on this problem.

Gravitational lensing by a single elliptical galaxy has been investigated in
\cite{Keeton_2000}, \cite{FassKeetonKhav}, \cite{KhavLund}, \cite{BergErem}.
As pointed out in \cite{FassKeetonKhav} and \cite{Pet10}, the image-counting problem for lensing by multiple elliptical galaxies represents an interesting uncharted territory.

The image-counting problem in gravitational lensing shares some mathematical similarities with other problems of mathematical physics, such as Maxwell's problem on the number of equilibria of electrostatic point charge systems, or the problem of determining relative equilibria in the circular-restricted $n$-body problem, see the introduction of \cite{ACLPZ} which elaborates on these similarities and provides some references to relevant work.

Additionally, the image-counting problem is related to a complex analytic problem posed by Sheil-Small \cite{SS} and refined by Wilmshurst \cite{W2} asking to determine the maximum number of zeros of complex harmonic polynomials, i.e., polynomials of the form $p(z) + \overline{q(z)}$ where $z$ is a complex variable, $p$ and $q$ are univariate polynomials of different degrees, and $\overline{q(z)}$ denotes the complex conjugate of $q(z)$, see \cite{KN2} for an expository paper spanning both topics.
This topic has been studied in several papers including \cite{W2},  \cite{BHS1995},  \cite{KhSw}, \cite{BL2004},  \cite{G2008},
\cite{LLL},
\cite{KhavLeeSaez},
\cite{HLLM},
\cite{SeteZur},
\cite{LundRand}, but the problem of determining the maximum number of zeros (in terms of the degrees of $p$ and $q$) is yet another challenging problem that remains open.

\subsection*{Outline of the paper} We review some preliminary results in Section \ref{sec:prelim} (Rhie's single-plane examples and a stability result from differential topology) that serve as important tools in the proof of Theorem \ref{thm:examples}. We prove Theorem \ref{thm:examples} in Section \ref{sec:main} where we first present the construction in the two-plane case for clarity before proving the general case.
We present the results of some relevant numerical simulations in Section \ref{sec:numerics}, where we also discuss the ``caustic of multiplicity'' phenomenon mentioned in the abstract.  The paper also contains an appendix that includes a discussion of the parameters $\e_i$ and the physical meaning of choosing each of them to be small (as will be done in our construction).  The appendix also includes an elementary result from matrix theory needed in the proof of Theorem \ref{thm:examples}.

\section{Preliminaries}\label{sec:prelim}

\subsection{Review of Rhie's single-plane extremal examples and their non-degeneracy} \label{sec:Rhie1}


S.\ H.\ Rhie constructed extremal single-plane examples with $g \geq 2$ point masses lensing $5g-5$ images of a single background source \cite{Rhie}.

We summarize those examples here while following the presentation in \cite{Bleher}.

Using complex variable notation $z=x+iy \in \C \cong \R^2$, for each 
$g\geq 4$, consider the lensing equation
\be\label{eq:Rhie}
z - \sum_{k=1}^{g-1} \frac{z-\zeta_{k,g}}{|z-\zeta_{k,g}|^2} - b^2\frac{z}{|z|^2} = 0,
\ee
which describes a gravitational lens with $g-1$ point-masses located at the vertices $\zeta_{k,g}=a e^{\frac{2\pi k}{g-1}i}$ of a regular polygon along with a point-mass with mass $b^2$ positioned at the origin.  For $g=2,3$ the construction is simpler; we omit the mass at the origin and the lens consists of equal point masses at $\zeta_{k,g+1} = e^{\frac{2\pi k}{g}i}$, $k=0,1$ for $g=2$ and $k=0,1,2$ for $g=3$.
We write the lens equation for these cases collectively as
\be\label{eq:Rhie2}
z - \alpha(z) = 0,
\ee
where $\alpha(z):= \sum_{k=0}^{g-1} \frac{z-\zeta_{k,g}}{|z-\zeta_{k,g}|^2} + b^2\frac{z}{|z|^2}$ when $g \geq 4$, and
$\alpha(z) := \sum_{k=0}^{g} \frac{z-\zeta_{k,g+1}}{|z-\zeta_{k,g+1}|^2}$ for $g=2,3$.

Choosing $a=(g-2)^{-1/(g-1)} \left(\frac{g-2}{g-1}\right)^{1/2}$, and choosing $b=b(g)>0$ sufficiently small, the system \eqref{eq:Rhie2} has $5g-5$ nondegenerate equilibria \cite[Proof of Prop. 5.2]{Bleher}.
It then follows from Lemma \ref{lemma:stable} that the same statement holds for a background source $w$ sufficiently close to $w=0$.  We state this as a remark.

\begin{rem}\label{rmk:delta}
For $\delta>0$ sufficiently small, the disk $D_\delta = \{w\in \R^2 : |w|< \delta\}$ is contained in the set
\be
\mathcal{U} = \{ w \in \R^2 : x - \alpha(x) = w \text{ has } 5(g - 1) \text{ nondegenerate solutions} \}.    
\ee
\end{rem}

We note that the statement in the remark is also a consequence of \cite[Prop. 5.2]{Bleher}.

\subsection{Nondegeneracy and stability}

We will need the following lemma which is an instance of the  transversality and stability principle from differential topology.
The lemma follows as a special case of Thom's isotopy Lemma \cite[Prop. 11.1]{Mather}.

\begin{lemma}\label{lemma:stable}
Let $\Omega \subset \R^d$ be a bounded domain, and let $F: \Omega \rightarrow \R^d$ be a smooth map that extends to be smooth in a neighborhood of the closure of $\Omega$.  Suppose the zero set $\{F=0\}$ is finite and nondegenerate, i.e., the Jacobian determinant of $F$ is nonvanishing at each point in the preimage $F^{-1}(0)$.
Then, there exists $\e>0$ such that for all smooth functions $\hat{F}: \Omega \rightarrow \R^d$ satisfying
$\left\| F - \hat{F} \right\|_{C^1(\Omega)} < \e$ the zero set $\{ \hat{F} = 0 \}$ is nondegenerate and has the same number of points as the the zero set $\{ F=0 \}$.
\end{lemma}

Here, $\left\| F \right\|_{C^1(\Omega)} := \sup_{x \in \Omega} |F(x)| + \sup_{x\in \Omega} \max_{1 \leq i,j \leq d} | J[F]_{i,j}(x) |$ denotes the $C^1$-norm.

\begin{rem}\label{rem:finiteness}
We note that the finiteness condition on the zero set $\{F=0\}$ actually follows as a consequence of nondegeneracy (along with the other assumptions in the lemma---that $\Omega$ is bounded and $F$ extends to be smooth in a neighborhood of the closure of $\Omega$).  Indeed, it follows from compactness of the closure of $\Omega$ that if $\{F=0\}$ were infinite then it would have an accumulation point where $F$ vanishes (by continuity), but nondegenerate zeros are isolated (which follows from the inverse function theorem).
\end{rem}

\section{Construction of multiplane ensembles}\label{sec:main}

In this section we present the construction of examples verifying the statement in Theorem \ref{thm:examples}.

\subsection{Two-plane examples}\label{sec:K=2}

To make the main ideas clear, let us first present the proof of Theorem \ref{thm:examples} in the two-plane case with $g_1$ masses in the first plane and $g_2$ masses in the second plane.

The construction will require scaling the parameters related to the second plane in order to ensure its solutions all lie within a disk of a certain radius. For this, we will use the following lemma.
\begin{lemma} \label{lem:1scale}
Fix $\lambda >0$,
and consider a single-plane lens equation 
$y=x - \alpha(x)$ with parameters $y,b_i,\xi_i$.
Then
$x$ is a solution to $y=x - \alpha(x)$ if and only if $\lambda x$ is a solution to the single-plane lens equation with scaled parameters $\lambda y, \lambda b_i, \lambda \xi_i$, i.e., scaling the parameters leads to scaling the solution set by the same factor.
\end{lemma}
\begin{proof}[Proof of Lemma]
Multiplying the lens equation
$ \displaystyle y = x -  \sum b_i^2 \frac{ x-\xi_i} {|x - \xi_i|^2}$
by $\lambda$ produces an equivalent equation $ \displaystyle \lambda y = \lambda x - \lambda \sum b_i^2 \frac{ x-\xi_i} {|x - \xi_i|^2}$ which can be manipulated as follows so that it is of the form of a lens equations with scaled parameters and scaled input variable.
\begin{align}
\lambda y &=  \lambda \left( x - \sum b_i^2 \frac {x - \xi_i}{|x - \xi_i|^2} \right)\\
  &=  \lambda x - \sum \lambda b_i^2 \frac{x - \xi_i} {|x - \xi_i |^2} \\
  &=  \lambda x - \sum (\lambda b_i)^2 \frac{\lambda x - \lambda \xi_i} {|\lambda x - \lambda \xi_i |^2}.
\end{align}
We observe that the final equation
$ \displaystyle \lambda y = \lambda x - \sum (\lambda b_i)^2 \frac{\lambda x - \lambda \xi_i} {|\lambda x - \lambda \xi_i |^2}$
is indeed a single-plane lens equation where both the parameters and input variable have been scaled by $\lambda$, and from this the conclusion of the lemma follows.
\end{proof}

\begin{proof}[Proof of Theorem \ref{thm:examples} in the case $K=2$]
In the case $K=2$, the system \eqref{eq:systemgen} consists of two vector equations, which we write as
\be\label{eq:system}
\begin{cases}
x_2   &= x_1 - \alpha_1 (x_1) \\
y  &= x_2 + \e (x_2 - x_1) - \alpha_2(x_2)
\end{cases},
\ee
where $\e=\e_2>0$ will be chosen to be small as we explain below, and we collect the $\beta_i$ constants into the bending angle vectors $\alpha_i$ by absorbing them into the Einstein radius parameters $b_{j,\ell}$. See the Appendix (Section \ref{sec:param}) for a discussion on the parameters $\e_i$ and $\beta_i$ in terms of the choice of the position of each lensing plane.

We will use a stability argument to show that, for particular $\alpha_1$ and $\alpha_2$ and with $\e>0$ sufficiently small, the system \eqref{eq:system} has $5(g_1-1)5(g_2-1)$ nondegenerate solutions.

First, consider a system with $\e=0$ which we write as
\be\label{eq:system2}
\begin{cases}
f_1(x_1,x_2) &:= x_1 - x_2  -\alpha_1 (x_1) = 0\\ \quad f_2(x_2) &:= x_2  - y - \alpha_2(x_2) = 0
\end{cases}.
\ee
For the source position we take $y=0$.
We choose the position and mass parameters appearing in the deflection term $\alpha_1$ using a single-plane Rhie ensemble with $g_1 \geq 2$ point-masses as described in Section \ref{sec:Rhie1}.
Let $\alpha_2^*$ denote the deflection term corresponding to a Rhie ensemble with $g_2$ point masses.  The asterisk in $\alpha_2^*$ indicates that this is a preliminary choice; we will arrive at a choice for $\alpha_2$ once we appropriately scale the parameters in $\alpha_2^*$.
As stated in Remark \ref{rmk:delta}, for $\delta>0$ sufficiently small we have that the disk $D_\delta = \{x_2\in \R^2 : |x_2|< \delta\}$ is contained in the set
\be
\mathcal{U} = \{ x_2 \in \R^2 : x_1-\alpha_1(x_1)=x_2 \text{ has } 5(g_1 - 1) \text{ nondegenerate solutions} \}.     
\ee
Fix such a $\delta>0$. Let $R$ be the radius of a disk containing all of the $5(g_2 -1)$ solutions of the equation $x_2-\alpha_2^*(x_2)=0$. Let $\lambda = \delta/R$. With $\alpha_2^*$ written as
$$
\alpha_2^*(x) = \sum_{i=0}^{g_2}b_{2,i}^2 \frac{x - \xi_{2,i}}{|x - \xi_{2,i}|^2}
$$
we define $\alpha_2$ by 
$$
\alpha_2(x) = \sum_{i=0}^{g_2} (\lambda  b_{2,i})^2 \frac{x - \lambda\xi_{2,i}}{|x - \lambda\xi_{2,i}|^2}.
$$
Note that, by Lemma \ref{lem:1scale} and the choice of the scaling parameter $\lambda=\delta/R$, the second equation $f_2(x_2) = 0$ in \eqref{eq:system2} has $5(g_2-1)$ solutions that are all contained within the disk $D_\delta$ of radius $\delta$. Hence, by our choice of $\delta$, substituting each of these $5(g_2-1)$ values of $x_2$ into the first equation $f_1(x_1,x_2) = 0$, gives rise to $5(g_1 - 1)$ solutions. Thus, the system \eqref{eq:system2} has $5(g_1-1)5(g_2-1)$ solutions.

Next we check that the solution set of the system \eqref{eq:system2} is nondegenerate, i.e., the Jacobian determinant of the map $F : \R^4 \to \R^4$ defined by $F(x_1,x_2) = (f_1(x_1,x_2), f_2(x_2))$ does not vanish at any of these $5(g_1-1)5(g_2-1)$ solutions. This can be seen by observing that the Jacobian of this map is 
\[
 J[F] = \begin{pmatrix}
J_1 & -I \\\
0 & J_2
 \end{pmatrix},
\]
where $J_i$ is the $2 \times 2$ Jacobian matrix of the mapping $x_i \to x_i - \alpha_i(x_i)$, $I$ denotes the $2 \times 2$ identity matrix, and $0$ denotes the $2 \times 2$ zero matrix. From the block upper triangular structure of this matrix, it follows that 
$\det(J[F]) = \det(J_1)\det(J_2)$ 
(see Lemma \ref{lem:matrix} in the Appendix).
By the known nondegeneracy of Rhie's single-plane examples (see again Remark \ref{rmk:delta}), we have $\det (J_2)$ does not vanish at any of the $5(g_2-1)$ solutions of $f_2(x_2)=0$.  Moreover, each of these $5(g_2-1)$ solutions of $f_2(x_2)=0$ satisfies $x_2 \in D_\delta$.  By choice of $\delta>0$, for each $x_2 \in D_\delta$ we have that $\det(J_1)$ does not vanish at any of the $5(g_1-1)$ solutions of the equation $x_1- \alpha_1(x_1) = x_2$.
Hence, $\det(J[F]) = \det(J_1)\det(J_2)$ does not vanish at any of the $5(g_2-1)5(g_1-1)$ solutions of the system \eqref{eq:system2}. As desired, this shows that the solution set of the system \eqref{eq:system2} is nondegenerate.

Finally, we obtain the desired lensing system as a perturbation of the above nondegenerate system.  Lemma \ref{lemma:stable} states that 
nondegeneracy is a stable condition with respect to $C^1$-small perturbations. 
Before applying Lemma \ref{lemma:stable}, we first need to restrict the domain of the map $F$ which is smooth except when the coordinate $x_\ell$ coincides with a point mass position $\xi_{j,\ell}$, so we remove from the $\ell \thh$ lens plane a small closed disk centered at each point mass position $\xi_{j,\ell}$ (while choosing the radii sufficiently small to avoid removing any solutions of \eqref{eq:nonphysical}).  In order to ensure the compactness condition, we also remove the complement of a large disk from each coordinate plane (choosing the radius large enough again to avoid removing any solutions).  The resulting domain $\Omega$ has compact closure, and $F$ is smooth on $\Omega$ and extends to be smooth in a neighborhood of the closure of $\Omega$.
As we have shown above, the system \eqref{eq:nonphysical} is nondegenerate, and choosing $\e>0$ small, the perturbation $(0, \e(x_2-x_1))$ appearing in \eqref{eq:system} can be made sufficiently $C^1$-small (over the closure of $\Omega$), so that we may apply Lemma \ref{lemma:stable}.
Choosing such an $\e>0$ thus furnishes a point-mass gravitational lens with $5(g_1-1)5(g_2-1)$ images of a single background source, and this proves the theorem in the case of $K=2$ lens planes.
\end{proof}

\begin{rem}
We note that the preliminary system \eqref{eq:system2} with $\e=0$ does not in general correspond to a physical gravitational lens, (see \cite[pg 288]{PetBook}).  The point of the above argument (which is applied again in the general $K$-plane case below) is that the preliminary nonphysical system can be perturbed to a physically meaningful system while preserving nondegeneracy and the number of solutions.
\end{rem}

\subsection{Proof of Theorem \ref{thm:examples} in the general setting}

We now proceed to the proof by induction of Theorem \ref{thm:examples} in the general setting of $K$ planes. The strategy is broadly the same as in the case of $k=2$ planes and uses a stable perturbation argument that begins by constructing a nonphysical example where $\e=0$, i.e., a system of equations of the form
\be\label{eq:nonphysical}
\begin{cases}
x_2   &= x_1 - \alpha_1 (x_1)  \\
x_3   &= x_2 - \alpha_2 (x_2) \\ 
&\vdots \\
x_{K}  &= x_{K-1}  -  \alpha_{K-1}(x_{K-1}) \\
y  &= x_{K}  -  \alpha_{K}(x_{K})
\end{cases}.
\ee
The only complication in adapting the proof to the general case is in how we rescale parameters. This will need to be done repeatedly, starting with the parameters related to the $K\thh$ plane and working backward.  To verify at the inductive step that we can successfully rescale parameters, we will need the following more general version of Lemma \ref{lem:1scale}.

\begin{lemma}\label{lem:scaleGen}
Let $\lambda >0$.  Consider the system
\be\label{eq:jtoK}
\begin{cases}
x_{j+1}  &= x_{j} - \alpha_j (x_j)  \\
x_{j+2}  &= x_{j+1} - \alpha_{j+1} (x_{j+1}) \\ 
&\vdots \\
x_{K}  &= x_{K-1}  -  \alpha_{K-1}(x_{K-1}) \\
y  &= x_{K}  -  \alpha_{K}(x_{K})
\end{cases}.
\ee
A vector $(x_j,x_{j+1},...,x_{K})$ solves the system of equations \eqref{eq:jtoK} with parameters $y,b_i,\xi_i$ 
if and only if the scaled vector $(\lambda x_j, \lambda x_{j+1},...,\lambda x_{K})$
solves the system \eqref{eq:jtoK} with scaled parameters  $\lambda y,\lambda b_i, \lambda \xi_i$,
i.e., scaling the parameters leads to scaling the solution set by the same factor.
\end{lemma}

\begin{proof}[Proof of Lemma]
Multiplying each equation in the system \eqref{eq:jtoK} by $\lambda>0$ we obtain an equivalent system of equations $\lambda x_{\ell+1} = \lambda x_\ell - \lambda \sum  b_i^2 \frac{ x_\ell-\xi_i} {|x_\ell - \xi_i|^2}$ for $ \ell=j,j+1,...,K$, where $x_{K+1}=y$.
Each of these equations can be manipulated as in the proof of Lemma \ref{lem:1scale} so that they each take the form 
$$\lambda x_{\ell+1} = \lambda x_\ell - \sum  (\lambda b_i)^2 \frac{ \lambda x_\ell - \lambda \xi_i} {|\lambda x_\ell - \lambda \xi_i|^2}.$$
We observe that this results in a system of the form \eqref{eq:jtoK} where all parameters and input variables have been scaled by $\lambda$, and this verifies the lemma.
\end{proof}


Before presenting the detailed proof, let us sketch the idea of the scaling algorithm, which may be informally described as follows: Taken as a single plane, the $K\thh$ plane lenses $5(g_K-1)$ images, and we imagine these as playing the role of sources for the $(K-1)\thh$ plane. Rescaling the parameters in the $K\thh$ plane (exactly as we did in the proof of the case $K=2$), each of these $5(g_K-1)$ sources gives rise to $5(g_{K-1}-1)$ images lensed by the $(K-1)\thh$ plane (treating it as a single plane at this step).  Using Lemma \ref{lem:scaleGen}, we then scale the parameters in both the $(K-1)\thh$ plane and the $K\thh$ plane, so that the $5(g_{K-1}-1)5(g_K-1)$ solutions of this restricted system each give rise, treating them as sources for the $(K-2)\thh$ plane, to $5(g_{K-2}-1)$ images, and so on. Repeatedly rescaling in this way until reaching the first plane eventually produces the desired (preliminary) system of the form \eqref{eq:nonphysical} with $\prod_{i=1}^K 5(g_i -1)$ solutions.

\begin{proof}[Proof of Theorem \ref{thm:examples}]
We now extend the argument from Section \ref{sec:K=2} to construct $K$-plane examples that produce $\prod_{i=1}^K (5g_i-5)$ images of a single background source. We write the system \eqref{eq:systemgen} as
\be\label{eq:systemK}
\begin{cases}
x_2   &= x_1 -\alpha_1 (x_1)  \\
x_{i+1}  &= x_i + \e_i (x_i - x_{i-1}) - \alpha_i(x_i), \quad i=2,3,...,K
\end{cases},
\ee
where $x_{K+1} = y$, the fixed source location, and as before we have collected the $\beta_i$ constants into the bending angle vectors $\alpha_i$ by absorbing them into the Einstein radius parameters $b_{j,\ell}$. 

Each of the parameters $\e_i>0$ will be chosen sufficiently small (the physical meaning of choosing $\e_i$ small is explained in the Appendix), hence viewing \eqref{eq:systemK} as a perturbation of the following system.
\be\label{eq:system2K0}
\begin{cases}
\quad f_1(x_1,x_2) &:= x_1 - x_2  -\alpha_1 (x_1) = 0 \\
\quad f_2(x_2,x_3) &:= x_2 - x_3  -\alpha_2 (x_2) = 0 \\
\quad\quad\quad \dots \\
f_{K-1}(x_{K-1},x_K) &:= x_{K-1} - x_{K}  -\alpha_{K-1} (x_{K-1}) = 0 \\
\quad \quad f_K(x_K) &:= x_K  - y - \alpha_K(x_K) = 0
\end{cases}.
\ee
The key advantage of working with this system is its triangular structure; the variables $x_j$ with $j<i$ are absent from the $i^\text{th}$ equation in the system. Consequently, the solution set of \eqref{eq:system2K0} can be described by back-substitution.

We must prove that, for an appropriate scaling of the parameters in the Rhie-type ensembles appearing within each plane, the system \eqref{eq:system2K0} has $\prod_{i=1}^K (5g_i-5)$ nondegenerate solutions. As before, given masses $g_1$, $g_2$, ... , $g_K$ we start with a preliminary choice of parameters for each plane based on the single-plane construction in Section \ref{sec:Rhie1}. This produces a system of equations: 

\be\label{eq:sys2K1}
\begin{cases}
\quad f_1(x_1,x_2) &:= x_1 - x_2  -\alpha_1 (x_1) = 0 \\
\quad f_2^{(1)*}(x_2,x_3) &:= x_2 - x_3  -\alpha_2^{(1)*}(x_2) = 0 \\
\quad\quad\quad \dots \\
f_{K-1}^{(K-2)*}(x_{K-1},x_K) &:= x_{K-1} - x_{K}  -\alpha_{K-1}^{(K-2)*} (x_{K-1}) = 0 \\
\quad \quad f_K^{(K-1)*}(x_K) &:= x_K  - y- \alpha_K^{(K-1)*}(x_K) = 0
\end{cases},
\ee
where an asterisk indicates the need for scaling. The number in parenthesis preceding the asterisk indicates how many times in the course of our construction the associated parameters will need to be scaled before arriving at the final choice.  When these numbers are all $0$, the ``countdown'' is complete, and we will arrive at our choice of parameters for the $\e =0$ ensemble.

\textit{Step 1:} The last equation $f_{K}^{(K-1)*}(x_K)=0$ in system \eqref{eq:sys2K1} has $5(g_K-1)$ solutions that all lie within a disk of some radius $R_K$ centered at the origin. As in Remark \ref{rmk:delta} given $\delta_K>0$ sufficiently small the disk  $D_{\delta_{K}}:=\{x_K \in \R^2 : |x_K|<\delta_K \}$ is contained in $\mathcal{U}_{K-1}$, the set of values $x_K \in \R^2$ such that $f_{K-1}^{(K-2)*}(x_{K-1},x_K)=0$ has $5(g_{K-1}-1)$ nondegenerate solutions. Fix such a $\delta_K$ and let $\lambda_K = \delta_K / R_K$. 
With $\alpha_K^{(K-1)*}$ expressed as
$$
\alpha_K^{(K-1)*}(x) = \sum_{i=0}^{g_K}b_{K,i}^2 \frac{x - \xi_{K,i}}{|x - \xi_{K,i}|^2}
$$
we define $\alpha_K^{(K-2)*}$ by scaling parameters
$$
\alpha_K^{(K-2)*}(x) := \sum_{i=0}^{g_2}(\lambda_K b_{K,i})^2 \frac{x - \lambda_K\xi_{K,i}}{|x - \lambda_K\xi_{K,i}|^2}.
$$
Likewise we define $$f_K^{(K-2)*}(x_K) := x_K - \lambda_K y - \alpha_K^{(K-2)*}(x_K) = 0.$$
Note that by Lemma \ref{lem:1scale} along with the choice of the scaling factor $\lambda_K$,
the equation $f_{K}^{(K-2)*}(x_K)=0$ has $5(g_K-1)$ solutions that all lie within the disk $D_{\delta_K}$, and hence the pair of equations $\{ f_{K-1}^{(K-2)*}(x_{K-1},x_K)=0, f_{K}^{(K-2)*}(x_K)=0\}$ have $5(g_K-1)5(g_{K-1}-1)$ solutions.

\textit{Step 2:} The values of $x_{K-1}$ in these $5(g_K-1)5(g_{K-1}-1)$ solutions are all contained in a disk of some radius $R_{K-1}$. For $\delta_{K-1}>0$ sufficiently small the disk  $D_{\delta_{K-1}}$ is contained in $\mathcal{U}_{K-2}$, the set of values $x_{K-1} \in \R^2$ such that $f_{K-2}^{(K-3)*}(x_{K-2},x_{K-1})=0$ has $5(g_{K-2}-1)$ nondegenerate solutions. Fix such a $\delta_{K-1}$ and let $\lambda_{K-1} = \delta_{K-1} / R_{K-1}$. We then scale the parameters associated to \textit{both} the ultimate and penultimate lens planes:

With $\alpha_{K-1}^{(K-2)*}$ written
$$
\alpha_{K-1}^{(K-2)*}(x) := \sum_{i=0}^{g_{K-1}}(b_{K-1,i})^2 \frac{x - \xi_{K-1,i}}{|x - \xi_{K-1,i}|^2}
$$
we define $\alpha_{K-1}^{(K-3)*}$ by 
$$
\alpha_{K-1}^{(K-3)*}(x) := \sum_{i=0}^{g_{K-1}}\lambda_{K-1}^2 (b_{K-1,i})^2 \frac{x - \lambda_{K-1}\xi_{K-1,i}}{|x - \lambda_{K-1}\xi_{K-1,i}|^2}
$$
and
$\alpha_{K}^{(K-3)*}$ by 
$$
\alpha_K^{(K-3)*}(x) = \sum_{i=0}^{g_K}(\lambda_K\lambda_{K-1})^2 (b_{K,i})^2 \frac{x - \lambda_K\lambda_{K-1}\xi_{K,i}}{|x - \lambda_K\lambda_{K-1}\xi_{K,i}|^2}
$$
Likewise we define 
$$f_{K-1}^{(K-3)*}(x_{K-1},x_K) := x_{K-1} - x_K - \alpha_{K-1}^{(K-3)*}(x_{K-1}) = 0$$
and 
$$f_{K}^{(K-3)*}(x_{K}) := x_{K} - \alpha_{K}^{(K-3)*}(x_{K}) = 0.$$ 

By Lemma \ref{lem:scaleGen} and our choice of $\lambda_K$ and $\lambda_{K-1}$, all $5(g_K-1)5(g_{K-1}-1)$ solutions all within the disk $D_{\delta_{K-1}}$.  Hence, the final three (updated) equations
$\{f_{K-2}^{(K-3)*}(x_{K-1},x_K)=0, f_{K-1}^{(K-3)*}(x_{K-1},x_K)=0, f_{K}^{(K-3)*}(x_{K})=0 \} $ have $5(g_K-1)5(g_{K-1}-1)5(g_{K-2}-1)$ solutions.

\textit{Step j:} Here we give the general $j\thh$ step in the construction, where $1<j\leq K-1$. We presume that we have, at the $j-1\thh$ step, scaled each ensemble in planes $K$, $K-1$, ... , $K-j+2$ such that for the restricted system of equations
\be\label{eq:restrsys} 
\begin{cases}
\quad f_{K-j+1}^{(K-j)*}(x_{K-j+1}, x_{K-j+2}) &= x_{K-j+1} - x_{K-j+2} - \alpha_{K-j+1}^{(K-j+1)*} = 0 \\
\\
\quad f_{K-j+2}^{(K-j)*}(x_{K-j+2}, x_{K-j+3}) &= x_{K-j+2} - x_{K-j+3} - \alpha_{K-j+2}^{(K-j+1)*} = 0 \\
 &\vdots \\
\quad f_{K-1}^{(K-j)*}(x_{K-1},x_K) &= x_{K-1} - x_{K}  -\alpha_{K-1}^{(K-j+1)*} (x_{K-1}) = 0 \\
\\
\quad \quad f_K^{(K-j)*}(x_K) &= x_K  - y- \alpha_K^{(K-j+1)*}(x_K) = 0
\end{cases},
\ee
there are $\prod_{\ell=0}^{j-1} 5(g_{K-\ell}-1)$ values in the coordinate $x_{K-j+1}$ corresponding to nondegenerate solutions, all lying within some radius $R_{K-j+1}$. For $\delta_{K-j+1}>0$ sufficiently small the disk  $D_{\delta_{K-j+1}}$ is contained in $\mathcal{U}_{K-j}$, the set of values $x_{K-1+1} \in \R^2$ such that $f_{K-j}^{(K-j-1)*}(x_{K-j},x_{K-j+1})=0$ has $5(g_{K-j}-1)$ nondegenerate solutions. Fix such a $\delta_{K-j+1}$ and let $\lambda_{K-j+1} = \delta_{K-j+1} / R_{K-j+1}$. We then scale the parameters associated to \textit{all} the lens planes from $K-j+1$ to $K$:

With $\alpha_{K-j+1}^{(K-j)*}$ written
$$
\alpha_{K-j+1}^{(K-j)*}(x) = \sum_{i=0}^{g_{K-j+1}}b^2_{K-j+1,i} \frac{x - \xi_{K-j+1,i}}{|x - \xi_{K-j+1,i}|^2}
$$
we define $\alpha_{K-j+1}^{(K-j-1)*}$ by 
$$
\alpha_{K-j+1}^{(K-j-1)*}(x) := \sum_{i=0}^{g_{K-j+1}}(\lambda_{K-j+1} b_{K-j+1,i})^2 \frac{x - \lambda_{K-j+1}\xi_{K-j+1,i}}{|x - \lambda_{K-j+1}\xi_{K-j+1,i}|^2},
$$

$\alpha_{K-j+2}^{(K-j-1)*}$ by 
$$
\alpha_{K-j+2}^{(K-j-1)*}(x) := \sum_{i=0}^{g_{K-j+2}}((\lambda_{K-j+2}\lambda_{K-j+1}) b_{K-j+2,i})^2 \frac{x - (\lambda_{K-j+2}\lambda_{K-j+1})\xi_{K-j+2,i}}{|x - (\lambda_{K-j+2}\lambda_{K-j+1})\xi_{K-j+2,i}|^2},
$$
and so on up to $\alpha_{K}^{(K-j-1)*}$ defined by
$$
\alpha_{K}^{(K-j-1)*}(x) := \sum_{i=0}^{g_K}\left(\left( \prod_{\ell=0}^{j-1}\lambda_{K-\ell} \right) b_{K,i}\right)^2 \frac{x - (\prod_{\ell=0}^{j-1}\lambda_{K-\ell})\xi_{K,i}}{|x - (\prod_{\ell=0}^{j-1}\lambda_{K-\ell})\xi_{K,i}|^2},
$$

We correspondingly define a system of $j+1$ scaled equations
\be\label{eq:restrsys2}
\begin{cases}
f_{K-j}^{(K-j-1)*}(x_{K-j}, x_{K-j+1}) &= x_{K-j} - x_{K-j+1} - \alpha_{K-j}^{(K-j-1)*}(x_{K-j}) := 0 \\
\\
 f_{K-j+1}^{(K-j-1)*}(x_{K-j+1}, x_{K-j+2}) &:= x_{K-j+1} - x_{K-j+2} - \alpha_{K-j+1}^{(K-j-1)*}(x_{K-j+1}) = 0 \\
 &\vdots \\
\quad f_{K-1}^{(K-j-1)*}(x_{K-1},x_K) &:= x_{K-1} - x_{K}  -\alpha_{K-1}^{(K-j-1)*} (x_{K-1}) = 0 \\
\\
\quad \quad f_K^{(K-j-1)*}(x_K) &:= x_K  - y- \alpha_K^{(K-j-1)*}(x_K) = 0
\end{cases},
\ee

By Lemma \ref{lem:scaleGen} and our choice of scaling parameters, all $\prod_{\ell=0}^{j-1} 5(g_{K-\ell}-1)$ solutions of system \eqref{eq:restrsys} are, after scaling, all within the disk $D_{\delta_{K-j+1}}$.  Hence, the system \eqref{eq:restrsys2} has $\prod_{\ell=0}^{j} 5(g_{K-\ell}-1)$ solutions.

Continuing in this fashion, after $K-1$ steps, we will have scaled the Einstein radii and mass locations in the $j\thh$ plane ensemble by $\prod_{\ell=2}^j \lambda_\ell$, where $j=2,...,K$. We then take $\alpha_j(x_j,x_{j+1}) := \alpha_j^{0*}(x_j,x_{j+1})$  to be the scaled ensembles for $i=2,...,K$. This results in a system of the form \eqref{eq:system2K0} that has $\prod_{\ell=1}^K 5(g_\ell-1)$ solutions.

Next we verify nondegeneracy,
that is, we show that the Jacobian determinant of the map $F : \R^{2K} \to \R^{2K}$ defined by 
\be \label{eqn:KF}
F(x_1,...,x_K) = (f_1(x_1,x_2), f_2(x_2,x_3),... , f_{K-1}(x_{K-1},x_K), f_K(x_K))
\ee
is nonvanishing at each point in the solution set.
From a direct computation we find that the Jacobian $\mathcal{J}[F]$ of the mapping $F$ is
\be\label{eq:Jacobian}
\mathcal{J}[F]=
\begin{bmatrix}
J_1 & -I & 0 & \cdots & 0 \\
0 & J_2 & -I & \ddots & 0 \\
0 & 0 & J_3 & \ddots & \vdots \\
\vdots & \vdots & \ddots & \ddots & -I \\
0 & 0 & \cdots & 0 & J_k
\end{bmatrix}
\ee
where $J_i = J_i(x_i)$ is the $2 \times 2$ Jacobian matrix of the mapping $x_i \rightarrow x_i - \alpha_i(x_i)$,
and $I$ denotes the $2 \times 2$ identity matrix.
The determinant $\det (\mathcal{J}[F])$ is the product of the determinants $\displaystyle \prod_{i=1}^K \det(J_i)$ (see Lemma \ref{lem:matrix} in Section \ref{sec:JND} of the Appendix), hence nondegeneracy of the solution set $\{ x \in \R^{2K}: F(x) = 0 \}$ follows from the nondegeneracy of the individual Rhie ensembles chosen in each lensing plane.

That the original system (\ref{eq:systemgen}) with a sufficiently small choice of each $\e_i>0$ has the same number of nondegenerate solutions is then a consequence of Lemma \ref{lemma:stable} (as in the two-plane case, in applying Lemma \ref{lemma:stable} we first restrict the domain of $F$ by removing from the $\ell \thh$ lens plane a small closed centered at each point mass position $\xi_{j,\ell}$ as well as removing the exterior of a large disk), and this completes the proof of the theorem.
\end{proof}

\section{Numerical results}\label{sec:numerics}

\begin{figure}
\includegraphics[width=0.8\textwidth]{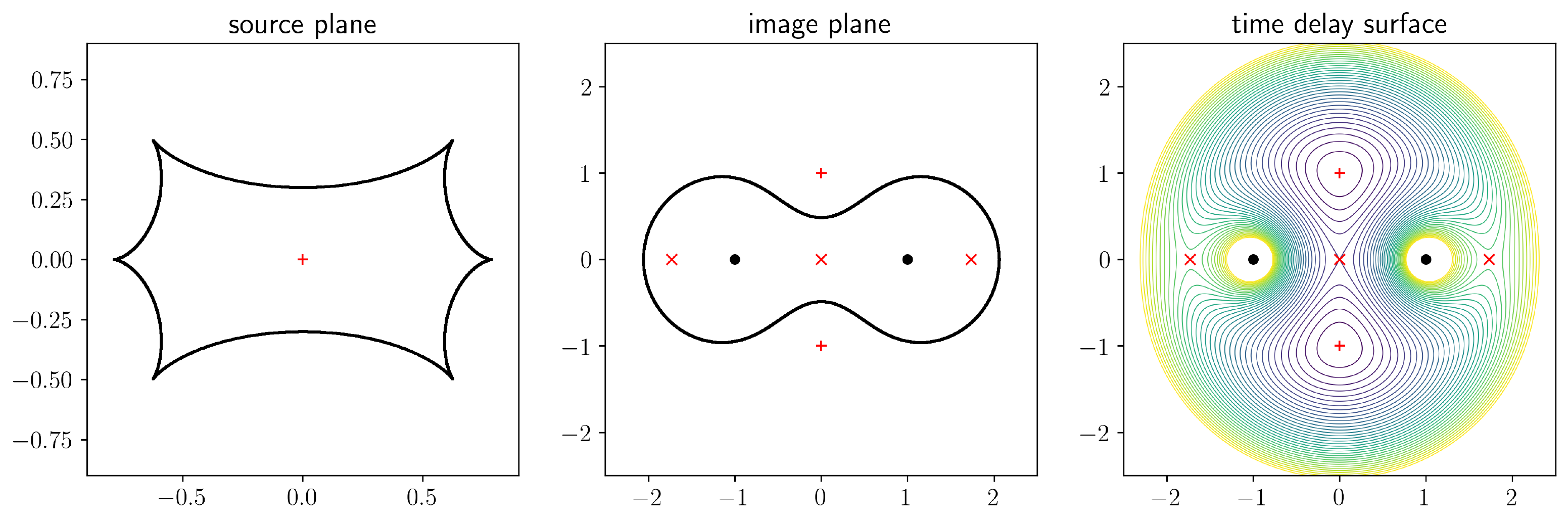}
\caption{
Example of a single plane with two masses with $b=1$.
The left panel shows the caustic in the source plane, and a source at the origin (red $+$).
The middle panel shows the critical curve in the image plane, along with the locations of the two masses (black dots), and the five lensed images: minima are indicated by $+$ and saddles by $\times$.
The right panel shows contours of the time delay function in the image plane, with the masses and images again indicated.
}\label{fig:single-g2}
\end{figure}

\begin{figure}
    \includegraphics[width=0.8\textwidth]{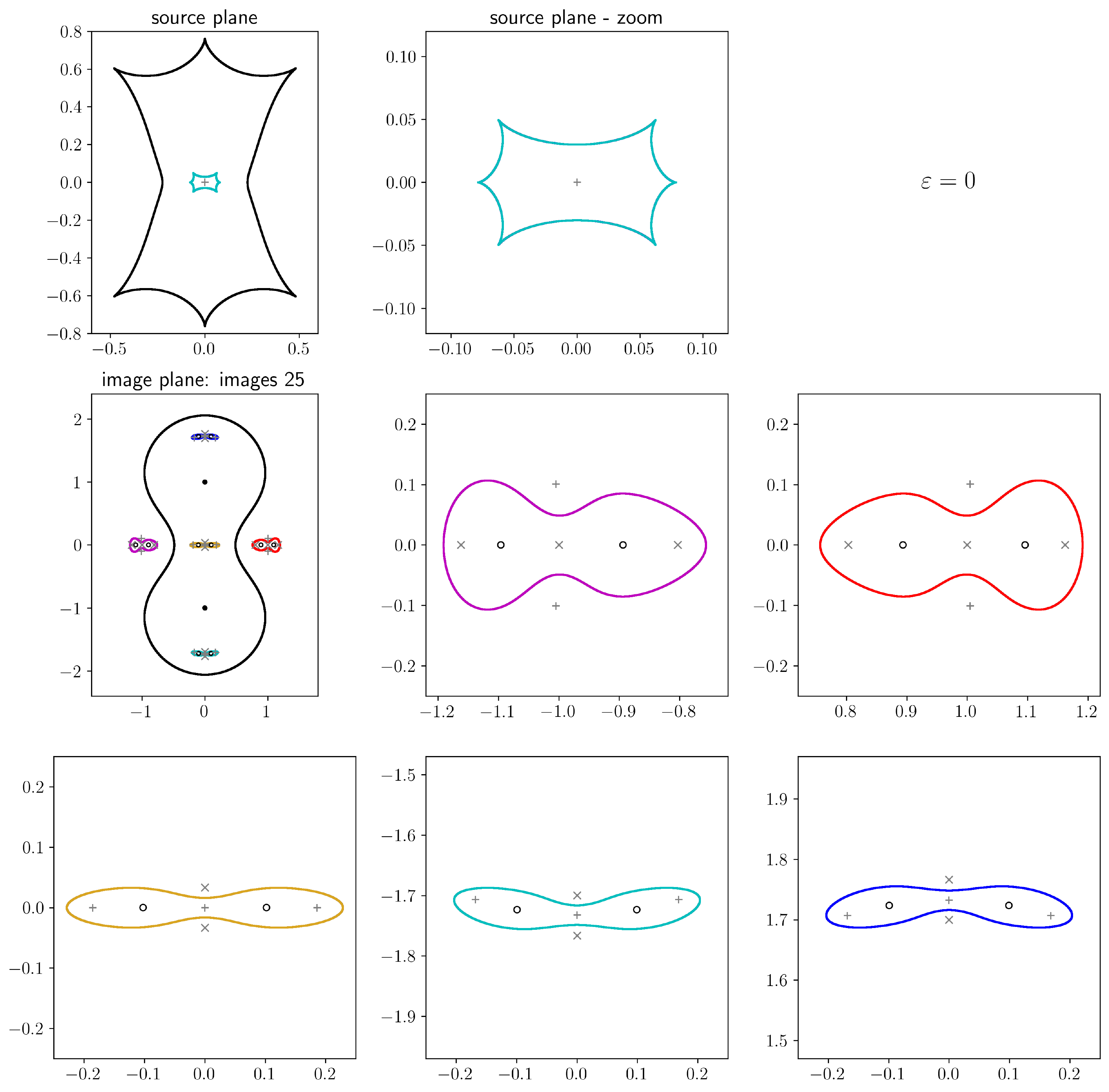}
\caption{
Example of a two-plane scenario with $g=2$ in each plane.
In the foreground plane, the two masses have $b=1$ and are placed on the $y$-axis.
In the background plane, the masses are placed on the $x$-axis, and the plane is scaled by $\lambda = 0.1$.
Here the two-plane scenario has $\e = 0$.
The top row shows the caustics in the source plane, along with the location of the source ($+$).
The middle left panel shows the full set of critical curves in the image plane, while the remaining panels show close ups of different areas.
Also shown are the 25 images, with minima indicated by $+$ and saddles by $\times$.
In the image planes, filled circles indicate the locations of the masses in the foreground plane, while open circles indicate projections of the masses in the background plane.
The small, central caustic (which appears cyan here) is actually an overlay of five copies of the caustic.
}\label{fig:double-g2-eps0}
\end{figure}

\begin{figure}
    \includegraphics[width=0.8\textwidth]{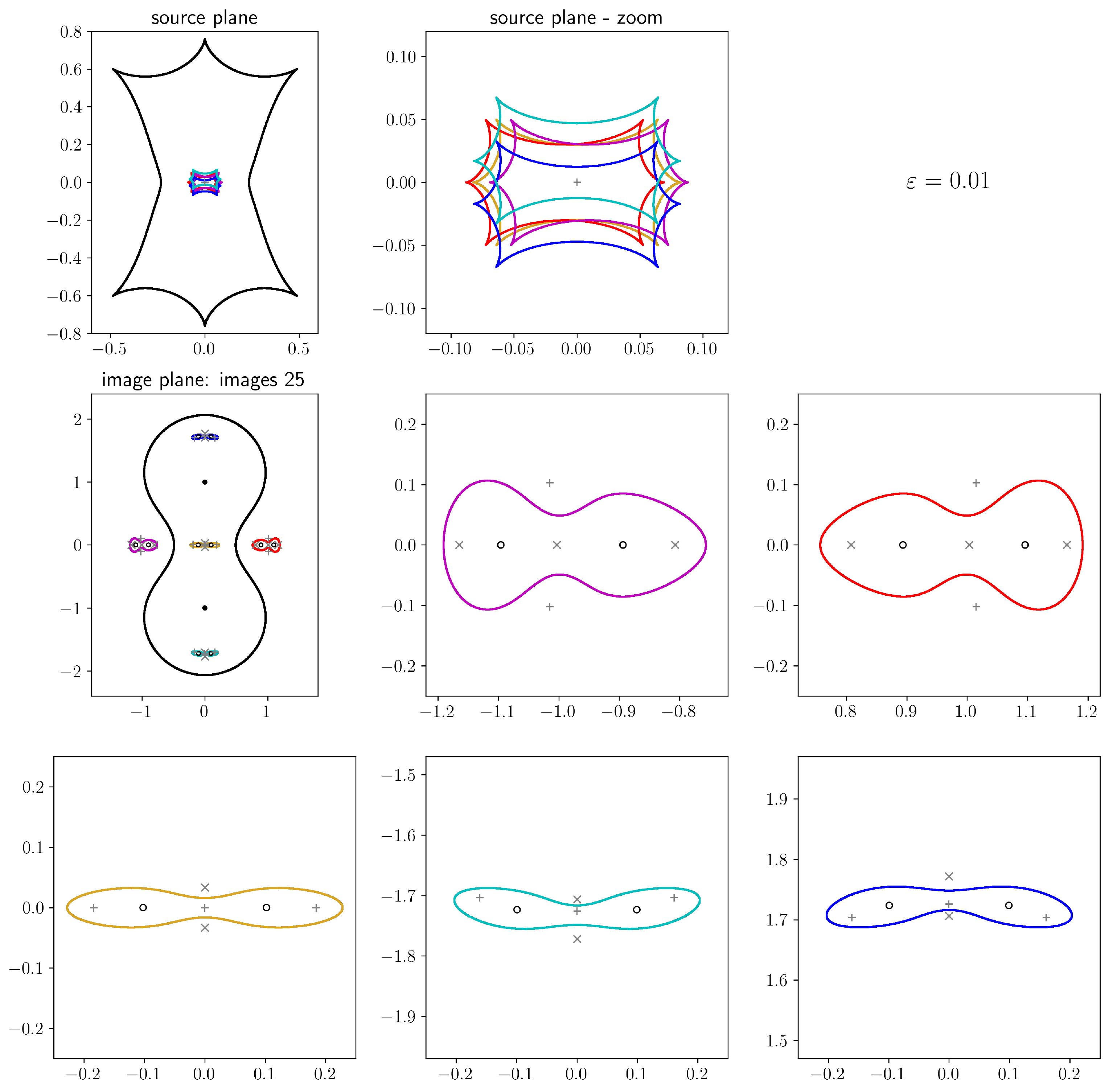}
\caption{
Similar to Fig.~\ref{fig:double-g2-eps0}, but now the two-plane scenario has $\e=0.01$.
Here it is apparent that the central caustic is in fact five distinct curves; the colors of the caustic curves match those of the corresponding critical curves.
}\label{fig:double-g2-eps0.01}
\end{figure}


\begin{figure}
    \includegraphics[width=0.8\textwidth]{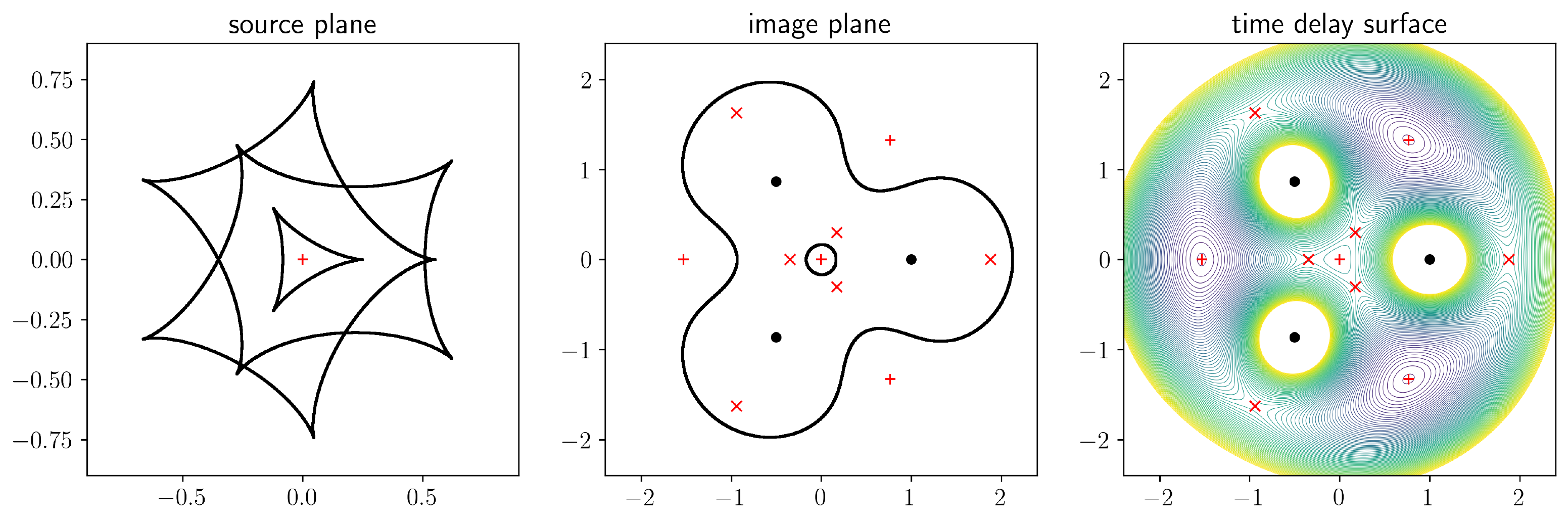}
\caption{
Similar to Fig.~\ref{fig:single-g2}, but for a Rhie ensemble with $g=3$.
The masses all have $b=1$ and are placed on the unit circle.
}\label{fig:single-Rhie3}
\end{figure}

\begin{figure}
    \includegraphics[width=0.8\textwidth]{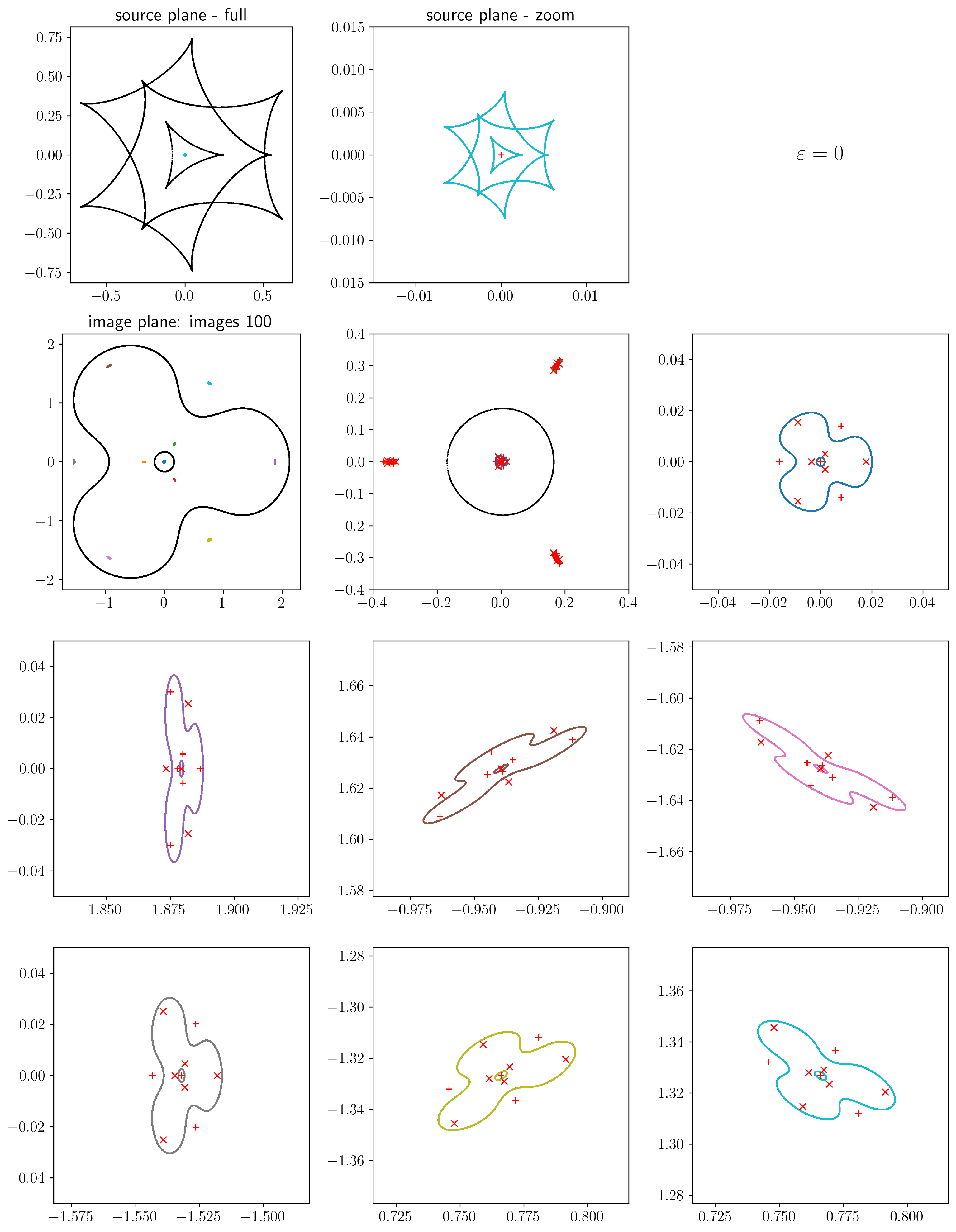}
\caption{
Results for a two-plane scenario with $g=3$ Rhie ensembles in each plane.
The background plane is scaled by $\lambda=0.01$.
Here the two-plane scenario has $\e=0$.
In the top row, the left panel shows the full caustics while the middle panel shows a close-up of the small caustics near the origin.
In the second row, the left panel shows the full critical curves, while the middle panel shows a close near the origin.
The remaining panels show further zooms centered on some (but not all) of the individual small critical curves.
Here there is a total of 100 lensed images comprised of 10 groups of 10.
}\label{fig:double-Rhie3-eps0}
\end{figure}

\begin{figure}
    \includegraphics[width=0.8\textwidth]{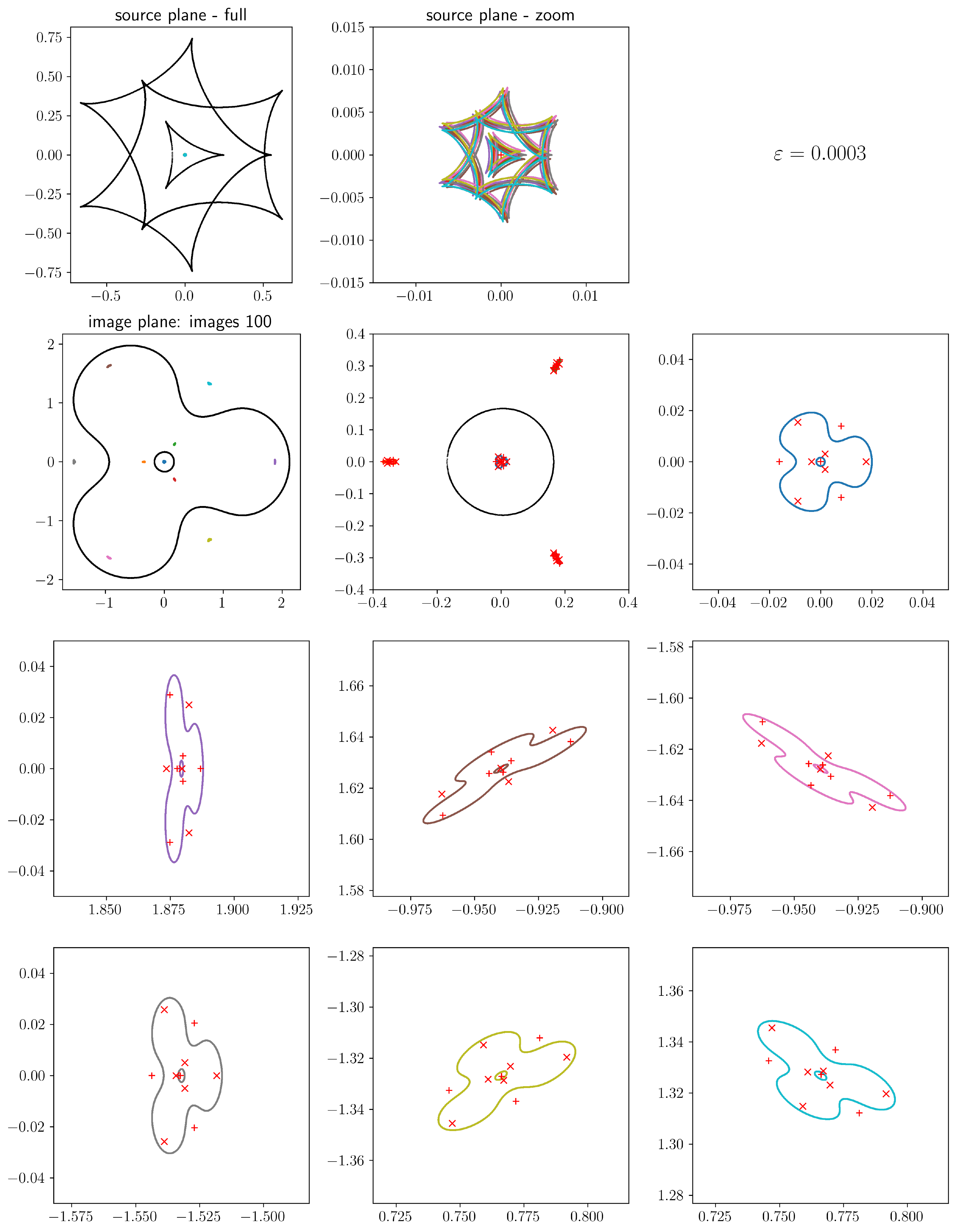}
\caption{
Similar to Fig.~\ref{fig:double-Rhie3-eps0}, but now the two-plane-scenario has $\e=0.0003$.
There are still 100 lensed images.
}\label{fig:double-Rhie3-eps0.0003}
\end{figure}

\begin{figure}
    \includegraphics[width=0.8\textwidth]{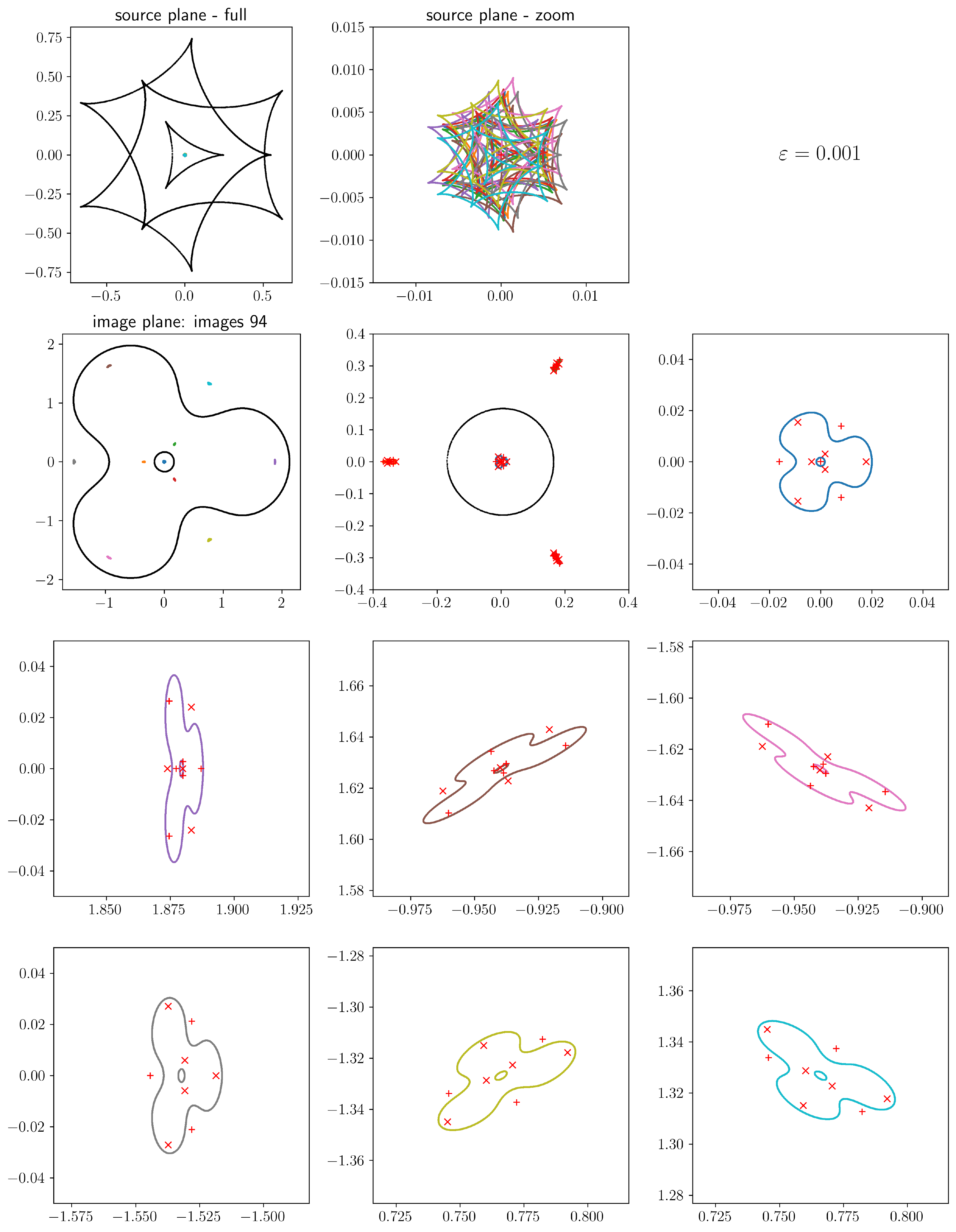}
\caption{
Similar to Fig.~\ref{fig:double-Rhie3-eps0}, but now the two-plane-scenario has $\e=0.001$.
This example has 94 lensed images.
In the bottom three panels, each group contains 8 images (compared with 10 in the other groups) because the associated caustic curve has shifted enough to make an image pair disappear.
}\label{fig:double-Rhie3-eps0.001}
\end{figure}

In this section, we present numerical results illustrating the above construction in some particular examples.

In addition to showing the positions of lensed images, we also display numerical plots of critical curves and caustics whose definitions we now recall (see \cite{PetBook} for a detailed exposition).
First, recall that the \emph{lensing map} $\eta: \R^2 \rightarrow \R^2$ associated to the system \eqref{eq:systemgen} takes a position $x_1$ in the first lens plane and maps it to a position $\eta(x_1)$ in the source plane obtained by tracing a light ray backward from observer to source (while accounting for the deflections in passing each lens plane).
Mathematically $\eta(x_1)$ can be expressed by solving for $x_2 = x_1 - \alpha(x_1)$ in the first equation in \eqref{eq:systemgen} and substituting this into the second equation, and then recursively expressing each $x_j$ in terms of $x_1$ and substituting into the next equation eventually expressing the final equation in the form $\eta(x_1)=y$, where recall $y$ is the source position.  Then the left hand side $\eta(x_1)$ of this final equation is the desired expression for the lensing map.  In the case of point-mass lenses and viewing $x_1 \in \R^2$ as a complex variable, $\eta(x_1)$ can be simplified to a rational expression in $x_1$ and its complex conjugate $\overline{x_1}$. We also note that in Figures \ref{fig:single-g2} and \ref{fig:single-Rhie3}, which illustrate single-plane scenarios, in addition to critical curves and caustics, we have also illustrated level sets of the time delay function $T:\R^2 \rightarrow \R$, which is defined in the single-plane case as $T(x):=|x-y|^2/2 - \sum_{i=1}^g b_i^2 \log|x-\xi_i|$ and has gradient satisfying $\nabla T(x) = \eta(x) - y$ (in particular, lensed images are critical points of $T$).

Now we can define the critical curve associated to a lensing map $\eta$ to be the vanishing set $\{x_1 \in \R^2 : \det J_\eta(x_1) =0 \}$ of the Jacobian determinant of $\eta$, and the caustic is defined as the set of critical values, i.e., the image of the critical curve under the lensing map $\eta$.  Note that the critical curve resides in the image plane (the first lens plane which is the same as the input space of the lensing map), while the caustic resides in the source plane (the target space of the lensing map).

For context, we first present in Figure \ref{fig:single-g2} a numerical simulation of a single-plane lens with $g=2$ point masses.
A source at the origin produces 5 images as expected.

In Figure \ref{fig:double-g2-eps0}, we present numerical simulation of a nonphysical $\e=0$ two-plane example with two masses in each plane, and in Figure \ref{fig:double-g2-eps0.01} we show the result of perturbing this to a physically meaningful system with $\e=0.01$.
Both cases produce the expected 25 images in 5 clusters with 5 images each.

The caustic structure of lensing maps in the nonphysical $\e=0$ case can exhibit an interesting feature where multiple components of the critical set are mapped to the same caustic (one may informally think of this as an overlay of multiple caustics).
This feature is quite striking when illustrating the perturbative construction used in the proof of Theorem \ref{thm:examples}.
Namely, comparing the caustics for $\e=0$ and $\e=0.01$, we see a ``caustic of multiplicity five'' that separates into five distinct caustics each winding around the origin.

This caustic-of-multiplicity phenomenon is unrelated to the symmetry in the Rhie ensembles; it is actually a typical feature of lensing maps in the nonphysical $\e=0$ case.  Indeed, when $\e=0$ the lensing map takes the form of a composition
$x_1 \mapsto \varphi_2 \circ \varphi_1(x_1)$, and by the multivariable chain rule its Jacobian matrix takes the form of a (matrix) product:
\be
J_{\varphi_2}(\varphi_1(x_1)) \cdot J_{\varphi_1}(x_1).
\ee
The critical set is where the determinant of the Jacobian vanishes
\be
\det \left[ J_{\varphi_2}(\varphi_1(x_1)) \cdot J_{\varphi_1}(x_1)\right] = \det J_{\varphi_2}(\varphi_1(x_1)) \cdot \det J_{\varphi_1}(x_1)=0.
\ee
The critical set is then the union of the two zero sets $\{\det J_{\varphi_1}(x_1)=0\}$ and $\{\det J_{\varphi_2}(\varphi_1(x_1)) =0\}$.  The latter is the same as the preimage $\varphi_1^{-1} \left\{ \det J_{\varphi_2}(x_2)=0 \right\}$.  Since the map $\varphi_1$ is many-to-one, a single component of the zero set $J_{\varphi_2}(x_2)=0$ may give rise to multiple components in the preimage $\varphi_1^{-1} \left\{\det  J_{\varphi_2}(x_2)=0 \right\}$, but all of those critical set components are mapped to a common caustic by the lensing map $\varphi_2 \circ \varphi_1$.  Indeed, in the composition $\varphi_2 \circ \varphi_1$ the map $\varphi_1$ is applied first and this trivially sends the preimage set $\varphi_1^{-1} \left\{ \det J_{\varphi_2}(x_2)=0 \right\}$ to the set $\{ \det J_{\varphi_2}(x_2)=0\}$.

Next we consider $g=3$. Figure \ref{fig:single-Rhie3} shows the single-plane case, which produces 10 images.

Figure \ref{fig:double-Rhie3-eps0} then shows a two-plane example with $g=3$ masses in each plane, for the unphysical case $\e=0$.
A source at the origin produces 100 images in 10 clusters of 10 images.

Figures \ref{fig:double-Rhie3-eps0.0003} and \ref{fig:double-Rhie3-eps0.001} then show the result of perturbing this to physically meaningful systems with different values of $\e$.
Once again we see the caustic-of-multiplicity phenomenon for $\e=0$.
When $\e>0$, the individual caustics separate. For $\e=0.0003$, the perturbation is small enough that the lens still achieves 100 images. However, for $\e=0.001$ three of the caustic curves have shifted enough that three pairs of images disappear, leaving a total of 94 images (see the bottom row of Fig.~\ref{fig:double-Rhie3-eps0.001}).

\section*{Author Declarations}
\noindent {\bf Data availability statement.}
The data that support this study are 
available upon request.

\noindent {\bf Conflict of interests.}
The authors have no conflicts to disclose.

\acknowledgements
This collaboration was initiated at the AMS MRC (Mathematics Research Community) workshop ``The Mathematics of Gravity and Light'' (Summer 2018).  The authors thank the AMS (American Mathematical Society) and also the National Science Foundation for supporting this program. The second named author is partly supported by the Simons Foundation, under the grant 712397.  We would also like to thank the anonymous referee for a careful reading and helpful suggestions that improved the presentation. We note that the author names on the title page of the current paper are ordered alphabetically following the convention in Mathematics.

\section{Appendix}

\subsection{A discussion of the parameters $\beta$ and $\e$} \label{sec:param}

The parameters $\beta_i$ and $\e_i$ in the lens equation can be expressed in terms of distances between lens planes. In Euclidean geometry, all of the distances here are simple Euclidean distances. For astrophysical applications, we must introduce cosmological distances (see, for example, \cite{Hogg}). In the standard cosmological model, the expanding universe is described by the Friedmann-Lema\^itre-Robertson-Walker metric with mass density $\Omega_M$, cosmological constant $\Omega_\Lambda$, curvature parameter $\Omega_k = 1 - \Omega_M - \Omega_\Lambda$, and current expansion rate $H_0$ (the Hubble constant). In what follows, we scale all distances by the Hubble distance $D_H = c/H_0$ to simplify the notation, and because we only need distance ratios for the lens equation.

The \emph{line-of-sight comoving} distance between redshifts $z_1$ and $z_2$ is given by
$$
  d^C(z_1,z_2) = \int_{z_1}^{z_2} \frac{dz}{\sqrt{\Omega_M(1+z)^3 + \Omega_k(1+z)^2 + \Omega_\Lambda}}
$$
The corresponding \emph{transverse comoving distance} $d^M$ is related to $d^C$ by
$$
  d^M = \begin{cases}
\frac{1}{\sqrt{\Omega_k}} \sinh\left(\sqrt{\Omega_k}\,d^C\right) & \Omega_k>0 \\
d^C & \Omega_k=0 \\
\frac{1}{\sqrt{|\Omega_k|}} \sin\left(\sqrt{|\Omega_k|}\,d^C\right) & \Omega_k>0
  \end{cases}
$$
Finally, the \emph{angular diameter distance} is
$$
  d^A(z_1,z_2) = \frac{d^M(z_1,z_2)}{1+z_2}
$$
To simplify the notation, we write this as $d^A_{1,2}$. If the first index is omitted, the first plane is taken to be the observer: $d^A_i = d^A(0,z_i)$.

The lens equation is naturally written in terms of angular diameter distances. The parameter $\beta_i$ is
$$
  \beta_i = \frac{d^A_{i,i+1}}{d^A_{i+1}}
  = \frac{d^M_{i,i+1}}{d^M_{i+1}}
$$
For the second equality, we note that the multiplicative redshift factors cancel, so $\beta_i$ can be written as a ratio of transverse comoving distances.

The parameter $\e_i$ is
$$
  \e_i = \frac{d^A_{j}\,d^A_{j-1,j+1}}{d^A_{j-1,j}\,d^A_{j+1}} - 1
  = \frac{d^M_{j}\,d^M_{j-1,j+1}}{d^M_{j-1,j}\,d^M_{j+1}} - 1
$$
(Note that $\e_1 = 0$.) Once again we note that the multiplicative redshift factors cancel so $\e_i$ can be written in terms of transverse comoving distances. Observational evidence suggests that our universe is spatially flat ($\Omega_k=0$). In such a universe, transverse comoving distances add in a simple way: $d^M_{i,j} = d^M_j - d^M_i$. Thus, if we assume a flat universe then we can simplify $\e_i$ to
$$
  \e_i = \frac{d^M_{i-1}\,d^M_{i,i+1}}{d^M_{i-1,i}\,d^M_{i+1}}
$$

Consider the special case of two lens planes. The first plane has
$$
  \e_1 = 0
  \qquad\mbox{and}\qquad
  \beta_1 = \frac{d^M_{1,2}}{d^M_2}
$$
The second plane has
$$
  \e_2 = \frac{d^M_1\,d^M_{2,3}}{d^M_{1,2}\,d^M_3}
  \qquad\mbox{and}\qquad
  \beta_2 = \frac{d^M_{2,3}}{d^M_3}
$$

\subsection{A note about scaling} \label{sec:scalings}

Including constants, the bending term that appears in the lens equation is
$$
  \beta_i \alpha_i(x_i) = \frac{d^M_{i,i+1}}{d^M_i} \sum_\ell \frac{G M_{i,\ell}}{c^2 d^A_i} \frac{x_i - \xi_{i,\ell}}{|x_i - \xi_{i,\ell}|^2}
$$
The factor $d^A_i$ is the angular diameter distance to plane $i$, and it serves to convert between the angular coordinates that are naturally used for the lens plane (namely $x_i$ and $\xi_{i,\ell}$) and the physical coordinates that enter the expression for the bending angle. The key scaling is
$$
  \frac{d^M_{i,i+1}\,M_{i,\ell}}{d^A_i d^M_{i+1}}
$$
For the two-plane scenario, the scalings for the two planes are
$$
  \frac{d^M_{1,2}\,M_{1,\ell}}{d^A_1 d^M_{2}}
  \qquad\mbox{and}\qquad
  \frac{d^M_{2,3}\,M_{2,\ell}}{d^A_2 d^M_{3}}
$$
In order to have $\e_2 \to 0$, we must have one of two cases:
\begin{itemize}

\item $d^M_1 \to 0$, which also implies $d^A_1 \to 0$ (since $d^A_1 = d^M_1/(1+z_1)$ where $z_1$ is the redshift of plane 1; note that $z_1 \to 0$ if the distances go to 0). In order for the bending term to remain constant, we must have $M_{1,\ell} \to 0$ such that $M_{1,\ell}/d^A_1 =$ const.

\item $d^M_{2,3} \to 0$. In order for the bending term to remain constant in this case, we must have $M_{2,\ell} \to \infty$ such that $d^M_{2,3}\,M_{2,\ell} =$ const.

\end{itemize}

\subsection{The determinant of a block upper triangular matrix} \label{sec:JND}

The following elementary result is most likely classical as it is a consequence (by a simple inductive proof) of the classical formula (see, for example, \cite{SylvesterBM}) 
\be\label{eq:detblock}
\det \left(\begin{bmatrix}
A & B \\
C & D\\
\end{bmatrix}\right) = \det(A - BD^{-1}C) \det(D)
\ee 
for the determinant of a block matrix, where $D$ is assumed to be invertible.
We include a proof of the lemma for the sake of completeness.

\begin{lemma} \label{lem:matrix}
Let $\mathcal{J}$ be an upper triangular block matrix, that is a matrix of the form 
$$
\mathcal{J}=
\begin{bmatrix}
J_1 & I_{1,2} & I_{1,3} & \cdots & I_{1,k} \\
0 & J_2 & I_{2,3} & \cdots & I_{2,k} \\
0 & 0 & J_3 & \cdots & I_{3,k} \\
\vdots & \vdots & \vdots & \ddots &\vdots \\
0 & 0 & 0 & \cdots & J_k 
\end{bmatrix}
$$
where each $J_i$ is an invertible $n_i\times n_i$ matrix, $I_{i,j}$ is an arbitrary $n_i \times n_j$ matrix, and each $0$ is an appropriately sized matrix of 0's. Then the determinant of $\mathcal{J}$ satisfies
\be\label{eq:detblockJ}
\det(\mathcal{J}) = \prod_{i=1}^k \det(J_i).
\ee
\end{lemma}

\begin{proof}
Applying the formula \eqref{eq:detblock} for the determinant of a block matrix, we first obtain 
$$
\det \left( \begin{bmatrix}
J_1 & I_{1,2} \\
0 & J_2\\
\end{bmatrix} \right)
= \det (J_1 - I_{1,2} J_2^{-1} 0 ) \det(J_2) = \det(J_1) \det(J_2),
$$
which establishes the base case $k=2$ for proving \eqref{eq:detblockJ} by induction on $k$, the number of diagonal blocks. Assume (for the inductive step) that the above formula is true for block matrices in the above form with less than $k$ diagonal blocks. Let $\mathcal{J}$ be as in the statement of the lemma, let 
$$
J=
\begin{bmatrix}
J_2 & I_{2,3} & \cdots & I_{2,k} \\
0 & J_3 & \cdots & I_{3,k} \\
\vdots & \vdots & \ddots &\vdots \\
 0 & 0 & \cdots & J_k 
\end{bmatrix},
$$
and let $I^* = [I_{1,2}, ... , I_{1,k}].$ Then we can write
$$
\mathcal{J} = 
\begin{bmatrix}
J_1 & I^* \\
0 & J 
\end{bmatrix},
$$
and another application of  \eqref{eq:detblock} gives $\det(\mathcal{J}) = \det(J_1 - I^*J^{-1} 0) \det(J) = \det(J_1)\det(J)$. By the inductive hypothesis we also have $\det(J)=\prod_{i=2}^k \det(J_i)$, and the desired result \eqref{eq:detblockJ} follows, completing the inductive step.
\end{proof}



\bibliography{grav}
\end{document}